\documentclass{article}

\usepackage{amsthm}
\usepackage{amsfonts,amsmath}
\usepackage{amssymb}
\usepackage{paralist}
\usepackage{multirow}
\usepackage{xspace}
\usepackage{graphicx}
\usepackage[table]{xcolor}
\usepackage{tikz}
\usepackage[textsize=scriptsize]{todonotes}
\usepackage{hyperref}
\usepackage[ruled, vlined]{algorithm2e}
\usepackage{setspace}
\newcommand{\setstretchvalue}{1.35}
\usepackage{colortbl}
\usepackage{lscape}
\usepackage{comment}
\usepackage{float}
\usepackage{tikz}
\usepackage{mathtools}
\usetikzlibrary{positioning}

\usepackage[T1]{fontenc}
\newtheorem{theorem}{Theorem}

\newtheorem{corollary}{Corollary}

\newtheorem{lemma}{Lemma}
\newtheorem{definition}{Definition}

\author{
  Danny Hermelin \\
  \emph{Ben-Gurion University of the Negev} \\
  \texttt{hermelin@bgu.ac.il}
  \and 
  Dvir Shabtay \\
  \emph{Ben-Gurion University of the Negev} \\
  \texttt{dvirs@bgu.ac.il}
  \and 
  Nimrod Talmon \\
  \emph{Weizmann Institute of Science} \\
  \texttt{nimrodtalmon77@gmail.com}
  }

\begin{document}

%\begin{frontmatter}

\title{On The Parameterized Tractability of the Just-In-Time Flow-Shop Scheduling Problem}

\maketitle

\iffalse
\author[add1]{Danny Hermelin}
\ead{hermelin@bgu.ac.il}
\author[add1]{Dvir Shabtay}
\ead{dvirs@bgu.ac.il}
\author[add2]{Nimrod Talmon\corref{cor1}}
\ead{nimrodtalmon77@gmail.com}

\cortext[cor1]{Corresponding author}

\address[add1]{Department of Industrial Engineering and Management, \\ Ben-Gurion University of the Negev, P.O.B. 653 Beer-Sheva 8410501, Israel}
\address[add2]{Faculty of Mathematics and Computer Science, \\ Weizmann Institute of Science, 234 Herzl Street, Rehovot 7610001 Israel}
\fi

\begin{abstract}
Since its development in the early 90's, parameterized complexity has been widely used to analyze the tractability of many NP-hard combinatorial optimization problems with respect to various types of problem parameters. While the generic nature of the framework allows the analysis of any combinatorial problem, the main focus along the years was on analyzing graph problems. In this paper we diverge from this trend by studying the parameterized complexity of Just-In-Time (JIT) flow-shop scheduling problems. Our analysis focuses on the case where the number of due dates is considerably smaller than the number of jobs, and can thus be considered as a parameter. We prove that the two-machine problem is W[1]-hard with respect to this parameter, even if all processing times on the second machine are of unit length, while the problem is in XP even for a parameterized number of machines. We then move on to study the tractability of the problem when combining the different number of due dates with either the number of different weights or the number of different processing times on the first machine. We prove that in both cases the problem is fixed-parameter tractable for the two machine case, and is W[1]-hard for three or more machines.
\end{abstract}

%\begin{keyword}
%  scheduling, flow-shop, just-in-time, parameterized complexity, fixed-parameter tractability
%\end{keyword}

%\end{frontmatter}

%\linenumbers

\section{Introduction}

The concept of just-in-time (JIT) production have attracted the attention of many industries in the last 50 years and has been wieldy adopted over the years to improve production efficiency in many industries (see, e.g., White and Prybutok~\cite{White} and Fullerton and McWatters~\cite{Fullerton}). One of the main concepts in JIT production systems is to make sure that customer orders (jobs to be produced) are completed exactly when required, avoiding unnecessary inventory and late delivery costs. Therefore, in the field of JIT scheduling, the objective to be minimized include penalties for both early and tardy completion of jobs. Two main types of such a cost functions are described in the literature. In the first, there is a penalty for each job that is proportional to the deviation of its completion time from the required due date. In the second, however, jobs that are not completed exactly at the due date incur a penalty which is independent of the deviation from the required due-date. Scheduling problems with a cost function of the first type are commonly known as \emph{earliness-tardiness scheduling problems} (see Baker and Scudder~\cite{Baker1990} for a survey), while problems with a cost function of the second type are commonly known as \emph{JIT scheduling problems} (see Shabtay and Steiner~\cite{DSGS2012} for a survey). It should be noted that the set of JIT scheduling problems on a single and parallel machines forms a special case of another important set of scheduling problems, which is the set of fixed interval scheduling problems (see Kovalyov \emph{et al.}~\cite{Kov2007} for a survey).

The JIT scheduling problem is solvable in polynomial time on a single machine (see Lann and Mosheiov~\cite{LannMosheiov}), and on various parallel machine systems (see, e.g., Arkin and Silverberg~\cite{ArkinSilverberg87}, Carlisle and Lloyd~\cite{Carlisle95} and {\v C}epek and Sung~\cite{Sung2005}). However, this is not the case when we move to more sophisticated scheduling systems such as flow-shop, job-shop and open-shop (see Choi and Yoon~\cite{Choi2007} and Shabtay and Bensoussan~\cite{DBLP:journals/scheduling/ShabtayB12}). The NP-hardness proof of the JIT scheduling problem on the later machine systems heavily depends on the fact that several parameters of the problem (such as the number of different processing times and the number of different due dates) may be arbitrary large in theory.

Such an assumption, however, may not be valid in many real-life scheduling problems. For example, in many cases the manufacturer produces only a predefined set of different products, yielding instances with limited number of different processing times. As another example, in many production systems the manufacturer limit the number of due dates such that each fits to one out of a predefined set of delivery dates, resulting instances with only limited number of different due dates. Therefore, it is quite natural to ask, if the NP-hard variants of the JIT scheduling problem becomes tractable when some of their natural parameters are of a limited size.

In this paper we study the parameterized tractability of the JIT scheduling problem in a flow-shop scheduling system. We do so by using the theory of parameterized complexity, which has been developed in the early 90's by the computer science community. The main idea in parameterized complexity is to analyze the tractability of NP-hard problems with respect to (wrt.) various instance parameters that may be independent of the total input length. Although the area of parameterized complexity has enjoyed tremendous success in many fields of combinatorial optimization since its development in the early 90's (see, e.g.,~\cite{Cygan2015,FG98,DF99,N06}), it is rarely used to analyze hard scheduling problems. In fact, we are aware of only handful number of papers which provide a parameterized analysis of scheduling
problems~\cite{BodlaenderandFellows1995,FellowsandMcCartin2003,hermelin2015scheduling,MnichandWiese2013,journals/scheduling/BevernMNW15,journals/corr/BevernNS15}.

Below, we continue the introduction by providing a brief exposition to the theory of parameterized complexity. Then, we formally define the scheduling problem we aim to analyze, survey the known results in the literature that are related to this problem, and present our research objectives.

\subsection{Basic Concepts in Parameterized Complexity Theory}

The main objective in parameterized complexity theory is to analyze the tractability of NP-hard problems wrt.\ their natural parameters.
Consider an NP-hard problem $\pi$ and let $n$ denotes its instance size.

\begin{definition}
Problem $\pi$ is \emph{fixed-parameter tractable (FPT)} wrt.\ some parameter $k$ if there is an algorithm that solves any instance of $\pi$ in $f(k)n^{o(1)}$ time,
for some computable function $f$ that solely depends on $k$.
An algorithm running in this running time is said to be a \emph{fixed-parameter algorithm}.
\end{definition}

\begin{definition}
Problem $\pi$ belongs to the \emph{XP set}, wrt.\ some parameter $k$ if there is an algorithm that solves any instance of $\pi$ in $n^{f(k)}$ time, for some computable function $f$ that solely depends on $k$.
An algorithm running in this running time is said to be an \emph{XP algorithm}.
\end{definition}

Note that a fixed-parameter algorithm is capable of solving instances where $k$ is fixed
(or upper bounded by a constant)
in polynomial time,
where,
importantly,
the degree of the polynomial does not depend on the value of the parameter $k$.
The main advantage of a fixed-parameter algorithm,
compared to an XP algorithm with time complexity $O(n^{k})$,
is that it separates the dependency between the instance size and the size of the parameter in the running time calculation,
making a fixed-parameter algorithm,
asymptotically speaking,
much more efficient.
There are parameterized analogues of NP-hardness which can be used to show that a problem is presumably not fixed-parameter tractable.
Similarly to the reductions used to show NP-hardness of some (non-parameterized) problem,
the standard technique for showing parameterized hardness is through a parameterized reduction.

\begin{definition}
A parameterized reduction from a problem $\pi_2$ to a problem $\pi_1$ is an algorithm mapping an instance $I_2$ of $\pi_2$ with parameter $k_2$
to an instance $I_1$ of $\pi_1$ with a parameter $k_1$ in time $f(k_2)|I_2|^{O(1)}$ such that $k_1\leq f(k_2)$ and $I_1$
is a yes-instance for $\pi_1$ if and only if $I_2$ is a yes-instance for $\pi_2$.
\end{definition}

There are several classes of problems that are conjectured not to be fixed-parameter tractable,
with the class of W[1]-hard problems and the class of W[2]-hard problems being the most popular.
Indeed,
for $i=1,2$,
a decision problem $\pi$ is W[i]-hard wrt.\ parameter $k$ if $\pi$ being fixed-parameter tractable wrt.\ $k$ infers that
all problems in W[i] are fixed-parameter tractable as well.

For proving that a problem is W[1]-hard one can provide a parameterized reduction from a known W[1]-hard problem,
such as the \textsc{Clique} problem, with $k$ being the size of the clique.
Similarly, for proving that a problem is W[2]-hard one can provide a parameterized reduction from a known W[2]-hard problem,
such as the \textsc{Set Cover} problem, with $k$ being the number of sets in the cover.
(while W[1] and W[2] are, formally speaking, different complexity classes, they both presumably rule out the existence of a fixed-parameter algorithm and for our point of view, similarly to the point of view of most papers studying parameterized complexity, the difference between them is of no special significance.)

\subsection{Problem Definition}

The aim of the current paper is to study the JIT flow-shop scheduling problem from a parameterized complexity perspective;
the problem can be defined as follows.
We are given a set of $n$ independent, non-preemptive jobs, $\mathcal{J}=\{J_1,J_2,\ldots,J_n\}$,
which are available for processing at time zero and are to be scheduled on a set of $m$ machines, $\mathcal{M}=\{M_1,M_2,\ldots,M_m\}$.
The machines are arranged in a flow-shop machine setting;
in such a setting all jobs are to be processed on all the machines and each job has to follow the same route trough the machines,
i.e.,
all jobs have to be processed first on $M_1$, then on $M_2$, and so on,
in a topological order up to $M_m$.

For $i=1,\ldots,m$ and $j=1,\ldots,n$, let $O_{ij}$ represents the processing operation of job $J_{j}$ on machine $M_{i}$ and let $p^{(i)}_{j}$ represents its processing time. We use $d_{j}$ to represent the due date of job $J_{j}$ and $w_{j}$ to represent the gain (income) of completing job $J_{j}$ in a JIT mode
(i.e., \emph{exactly} at time $d_{j}$).
We assume that all the $d_{j}$, $w_{j}$, and the $p^{(i)}_{j}$ values, are positive integers.

For a JIT scheduling problem,
a partition of the set $J$ into two disjointed subsets $E$ and $T$ is considered to be a feasible partition
(or a feasible schedule),
if it is possible to schedule jobs belonging to the set $E$ such that they are all completed in a JIT mode.
In a feasible schedule, without loss of generality, jobs belonging to the set $T$ are assumed to be rejected.
The objective is to find a feasible schedule with a maximal weighted number of jobs in the set $E$;
that is, to maximize $\sum_{J_{j}\in E} w_{j}$.

Following the classical three-field notation of~\cite{Graham79}, we denote our JIT flow-shop scheduling problem by $F_m || \sum_{J_{j}\in E} w_{j}$,
where $F_m$ in the first field indicates that the scheduling is done in a flow-shop scheduling system with $m$ machines
and the objective to be maximized appears in the last field.
Note that the middle (second) field is empty,
while in general it is used to specify specific processing characteristics and constraints.
By $F_m || |E|$ we refer to the special case where weights are identical,
and the problem is simply to maximize the number of JIT jobs.

\subsection{Related Work}

The first to study the $F_m || \sum_{J_{j}\in E} w_{j}$ problem were Choi and Yoon~\cite{Choi2007}.
By a reduction from the \emph{Partition} problem, they proved that the two-machine case is NP-hard.
They also show that the unweighted version of the problem is solvable in polynomial time (specifically, $O(n^4)$ time)
for two machines and that it is strongly NP-hard for three machines.
They left, however, an open question concerning whether the $F_2 || \sum_{J_{j}\in E} w_{j}$ problem is strongly NP-hard or ordinary NP-hard.
Shabtay and Bensoussan~\cite{DBLP:journals/scheduling/ShabtayB12} resolved this question by providing a pseudo-polynomial time algorithm for its solution
(meaning that the problem on two machines is ordinary NP-hard).
Shabtay and Bensoussan also show how the pseudo-polynomial time algorithm can be converted into a fully polynomial time approximation scheme
(FPTAS),
which provides $(1 + \varepsilon)$ approximation in $O((n^4/\varepsilon) \log(n^2/\varepsilon))$ time.
The latter result was later improved by Elalouf et al.~\cite{DBLP:journals/scheduling/ElaloufLT13},
who accelerated the running time of the FPTAS by a factor of $n\log (n/\varepsilon)$.

Shabtay~\cite{DBLP:journals/eor/Shabtay12} provided an $O(n^3)$ time algorithm for solving the $F_2 || |E|$ problem,
improving an earlier result of Choi and Yoon~\cite{Choi2007} by a factor of $n$.
Shabtay further showed that the problem on $m$ machines with machine-independent processing times is ordinary NP-hard in general and is solvable in $O(n^4)$ time
as long as all weights are identical.
If, on the other hand, the processing times are job-independent (but machine-dependent),
then the $m$ machine problem is solvable in polynomial time (specifically, $O(n^3)$).
This time complexity was pushed down to $O(n\log n)$ if in addition to job-independent processing times, it holds that all weights are identical.
Finally, Shabtay~\cite{DBLP:journals/eor/Shabtay12} show that the $m$ machine problem is solvable in $O(mn^2)$ time if there is a no-wait restriction,
i.e.,
if jobs are not allowed to wait between machines.

\subsection{Research Objective}

Our main research objective is to study the tractability of the hard variations of the $F_m || \sum_{J_{j}\in E} w_{j}$ problem
wrt.\ the number of different due dates, $\#d$.
In Section~\ref{section: prop} we provide some important properties of an optimal schedule.
In Section~\ref{section: pre} we prove that the two-machine problem is W[1]-hard wrt.\ this parameter,
even if all processing times in the second machine are of unit size.
We also prove that the more general $m$-machine problem belongs to XP,
for a parameterized number of machines.

We further continue to consider combining the number $\#d$ of different due dates (as a parameter) with two other parameters:
  the number of different processing times on the first machine ($\#p^1$)
  and
  the number of different weights ($\#w$).
In Section~\ref{section: pre2} we prove,
for both of these combined cases,
that the two-machine problem is fixed-parameter tractable;
and in Section \ref{section: pre3} we prove that the corresponding three-machine problem is W[1]-hard.
A summary and future research section concludes our paper (our results are also summarized in Table~\ref{table:results} below).

\section{General Properties}\label{section: prop}

Let $\sigma_i$ be the processing order of the set of JIT jobs (i.e., the set $E$)
on the machine $M_i$, for $i=1,\ldots,m$.
A job schedule is called a \emph{permutation schedule} if $\sigma_1=\sigma_2=\ldots=\sigma_m$.
The following two lemmas are due to Choi and Yoon~\cite{Choi2007}.

\begin{lemma}\label{lemma:Choi1}
  There exists an optimal schedule for the \textsc{$F_m||\sum_{J_{j} \in E} w_j$} problem in which $\sigma_1=\sigma_2$.
\end{lemma}

\begin{lemma}\label{lemma:Choi2}
 In an optimal schedule for the \textsc{$F_m||\sum_{J_{j} \in E} w_j$} problem,
 $\sigma_m$ follows the earliest due date (EDD) order.
\end{lemma}

It follows from Lemma~\ref{lemma:Choi1} and Lemma~\ref{lemma:Choi2}
that there exists an optimal permutation schedule for the \textsc{$F_2||\sum_{J_{j} \in E} w_j$} problem
in which both $\sigma_1$ and $\sigma_2$ follow the EDD rule.
When $m\geq 3$,
however,
the optimal schedule is not necessarily a permutation schedule;
indeed, Choi and Yoon~\cite{Choi2007} provide a counterexample showing that there exists a set of instances for the \textsc{$F_3||\sum_{J_{j} \in E} w_j$} problem
in which the optimal permutation in $M_1$ and $M_2$ does not follow the EDD rule.

Let $\sigma_i=(J_{\sigma_i(1)},J_{\sigma_i(2)},\ldots,J_{\sigma_i(n)})$.
Then,
the following lemma holds.

\begin{lemma} \label{lemma:sh}
  There exists an optimal schedule in which jobs in machines \linebreak
  $M_1,\ldots,M_{m-1}$ are scheduled as soon as possible,
  i.e.,
  where job $J_{\sigma_i(j)}$ starts on $M_i$ at the time which is the maximum between
  ($i$) the completion of job $J_{\sigma_i(j-1)}$ on $M_i$;
  and
  ($ii$) the completion of job $J_{\sigma_i(j)}$ on $M_{i-1}$.
\end{lemma}

\begin{table}[t]
\centering
\setstretch{1.35}
\begin{tabular}{ c | c c}
                 & \textsc{$F_2||\sum_{J_{j} \in JIT} w_j$}   & \textsc{$F_3||\sum_{J_{j} \in JIT} w_j$}    \\ \hline
$\#d$            & XP (Theorem~\ref{theorem:twotwo})          & XP (Corollary~\ref{corollary:one})          \\
$\#d + \#w$      & FPT (Theorem~\ref{theorem:four})           & W[1]-hard (Theorem~\ref{theorem:five})      \\
$\#d + \#p^1$    & FPT (Theorem~\ref{theorem:three})          & W[1]-hard (Theorem~\ref{theorem:hard2})     \\
$\#d + \#p^2$    & \multicolumn{2}{c}{W[1]-hard (Theorem~\ref{theorem:one})}
\end{tabular}
\caption{Summary of our results.
The number of different due dates is denoted by $\#d$,
the number of different processing times in the first (second) machine is denoted by $\#p^1$ ($\#p^2$),
and the number of different weights is denoted by $\#w$.}
\label{table:results}
\end{table}

\section{The Complexity of JIT Flow-Shop wrt.\ $\#d$}\label{section: pre}

In this section we consider the parameterized complexity of the JIT flow-shop with respect to a single parameter,
namely the number $\#d$ of different due dates.

\subsection{The W[1]-Hardness of the Two-Machines Case}

Consider the following definition of the \textsc{$k$SUM} problem.

\begin{definition}[$k$SUM]\label{definition:ksum}
Given
a set of $h$ integers $X=\{x_1, \ldots, x_h\}$,
an integer $k$ ($1\leq k <h$),
and an integer $B$,
determine whether there exists a set $S\subseteq X$ (possibly with repetitions) such that $|S| = k$ and $\sum_{x_i\in S}x_i = B$.
\end{definition}

It is known that the \textsc{$k$SUM} is W[1]-hard wrt.\ $k$ (see~\cite{downey1992fixed}).
Below we prove that the \textsc{$F_2||\sum_{J_{j} \in E} w_j$} problem is W[1]-hard wrt.\ $\#d$
by providing a parameterized reduction from the \textsc{$k$SUM} problem.

\begin{theorem}\label{theorem:one}
  The \textsc{$F_2||\sum_{J_{j} \in E} w_j$} problem is W[1]-hard when parameterized by the number $\#d$ of different due dates
  even if all jobs have unit processing times on the second machine.
\end{theorem}

\begin{proof}
Given an instance of \textsc{$k$SUM} we construct an instance for the decision version of the $F_2 || \sum_{J_{j}\in E} w_{j}$ problem as follows.
The set $\mathcal{J}$ includes $n=kh+1$ jobs.
The first $kh$ jobs are the union of $k$ sets $\mathcal{J}_1, \mathcal{J}_2,\ldots.,\mathcal{J}_k$,
where each set $\mathcal{J}_i$ ($i\in \{1,\ldots,k\}$) includes $h$ jobs.
Let $J_{ij}$ be the $j$th job in the set $\mathcal{J}_i$.
Moreover,
for any job $J_{ij}$,
let $d_{ij}$, $w_{ij}$, $p^{(1)}_{ij}$, and $p^{(2)}_{ij}$, be the due date, the weight, and the processing times on machines $M_1$ and $M_2$, respectively.

For any job $J_{ij}$,
we set
($i$) $d_{ij}=iT$, where $T=\sum_{x_i\in X}x_i$;
($ii$) $w_{ij}=T+x_j$;
($iii$) $p^{(1)}_{ij}=x_j$;
and
($iv$) $p^{(2)}_{ij}=1$.
For job $J_{kh+1}$,
we set
($i$) $d_{kh+1}=(k+1)T + 1$;
($ii$) $w_{kh+1}=k^2(T+1)^2$;
($iii$) $p^{(1)}_{kh+1}=(k+1)T - B$;
and
(iv) $p^{(2)}_{kh+1}=1$.
In our decision version of the $F_m || \sum_{J_{j}\in E} w_{j}$ problem
we ask whether there is a feasible schedule (partition) with $\sum_{J_{j}\in E} w_{j}\geq kT+B+k^2(T+1)^2$.
This finishes the description of the polynomial-time reduction
(a tabular representation of jobs created by the reduction is given in Table~\ref{table:reduction-1-general}).

\begin{table}[t]
\centering
\begin{tabular}{ c | c c c c}
            & $p^{(1)}$    & $p^{(2)}$    & $w$             & $d$            \\ \hline
$J_{11}$    & $x_1$        & $1$          & $T+x_1$         & $T$            \\
$\vdots$    & $\vdots$     & $\vdots$     & $\vdots$        & $\vdots$       \\
$J_{1h}$    & $x_h$        & $1$          & $T+x_h$         & $T$            \\ \hline
$\vdots$    & $\vdots$     & $\vdots$     & $\vdots$        & $\vdots$       \\ \hline
$J_{k1}$    & $x_1$        & $1$          & $T+x_1$         & $kT$           \\
$\vdots$    & $\vdots$     & $\vdots$     & $\vdots$        & $\vdots$       \\
$J_{kh}$    & $x_h$        & $1$          & $T+x_h$         & $kT$           \\ \hline
$J_{kh+1}$  & $(k+1)T - B$ & $1$          & $k^2(T+1)^2$    & $(k+1)T + 1$   \\
\end{tabular}
\caption{jobs created in the reduction described in the proof of Theorem~\ref{theorem:one}.}
\label{table:reduction-1-general}
\end{table}

Below we argue for the correctness of the reduction.
We begin by noticing that,
since all jobs in the set $\mathcal{J}_i$ ($i=1,\ldots,k$) share the same due date,
namely~$iT$,
it follows that,
in any feasible schedule,
at most one of them can be scheduled in a JIT mode.

Therefore, it holds that
\begin{equation}\label{eqn: cc1}
  |E|\leq k+1.
\end{equation}

Let us now prove that,
if we have a yes-instance for the \textsc{$k$SUM},
then there exists a feasible schedule for the constructed instance of the $F_2 || \sum_{J_{j}\in E} w_{j}$ problem
with $\sum_{J_{j}\in E} w_{j}\geq kT+B+k^2(T+1)^2$.
The fact that we have a yes-instance to \textsc{$k$SUM} implies that there exists a subset of $k$ elements
$S = \{x_{[1]},x_{[2]},\ldots,x_{[k]}\} \subseteq X$
such that $\sum_{j=1}^k x_{[j]} = B$,
where $[j]$ is the index of the $j$th element in $S$.
We then construct the following feasible schedule for the constructed instance of our scheduling problem.
We set $E=\{J_{1[1]},J_{2[2]},\ldots,J_{k[k]}\} \cup \{J_{kh+1}\}$ and $T=\mathcal{J} \backslash E$.
For $i=1,\ldots,k$, we then schedule job $J_{i[i]}$ during the time interval $(\sum_{j=1}^{i-1} x_{[j]}, \sum_{j=1}^{i} x_{[j]}]$ on $M_1$,
and during the time interval $(iT-1, iT]$ on $M_2$
(since $\sum_{j=1}^k x_{[j]} = B<T$, it follows that there is no overlap between the processing operations of the same job on both machines).
Moreover, we schedule job ${J_{kh+1}}$ during the time interval $(B, (k+1)T]$ on $M_1$
(which is just after the completion of job $J_{k[k]}$ on that machine)
and during the time interval $((k+1)T, (k+1)T+1]$ on $M_2$.
The constructed schedule is illustrated in Figure~\ref{figure:reduction-1}.
Since there is no overlap between the processing operations and all jobs in the set $E$,
and these jobs are scheduled in a JIT mode,
we have that the schedule is feasible;
further, it follows that
\begin{equation}\label{eqn: c1}
  \sum_{J_{j}\in E} w_{j}=\sum_{i=1}^{k} (T+x_{[i]})+k^2(T+1)^2=kT+B+k^2(T+1)^2;
\end{equation}
thus, the constructed instance of our scheduling problem is a yes-instance.

Now we prove that,
if we have a no-instance for the \textsc{$k$SUM},
then a feasible schedule for the constructed instance of the $F_2 || \sum_{J_{j}\in E} w_{j}$ problem
with $\sum_{J_{j}\in E} w_{j}\geq kT+B+T^2(k+1)^2$ does not exist.
By contradiction,
assume that a feasible partition (schedule) $\tau = E\cup T$ with $\sum_{J_{j}\in E} w_{j}\geq kT+B+T^2(k+1)^2$ exists.
The fact that at most a single job from each $\mathcal{J}_i$ can be scheduled in a JIT mode in any feasible schedule
implies that the total weight of the subset of JIT jobs among all jobs in $\cup_{i=1}^{k} \mathcal{J}_i$ is at most $kT+T=T(k+1)$.
Therefore, ${J_{kh+1}}$ is included in set the $E$
(as otherwise, we have that $\sum_{J_{j}\in E} w_{j} \leq T(k+1)<kT+B+T^2(k+1)^2$,
contradicting our assumption that $\sum_{J_{j}\in E} w_{j}\geq kT+B+T^2(k+1)^2$).
Let $E'=E \setminus \{J_{kh+1}\}$.
Since $\sum_{J_{j}\in E} w_{j}\geq kT+B+T^2(k+1)^2$ and ${J_{kh+1}}\in E$,
we conclude that
\begin{equation}\label{eqn: a1}
  \sum_{J_{j}\in E'} w_{j} = \sum_{J_{j}\in E} w_{j} - w_{kh+1} \geq kT+B.
\end{equation}

Let us now prove that $|E'|=k$.
The fact that $|E'|\leq k$ follows from eq.~(\ref{eqn: cc1}) above and the fact that $|E'|=|E|-1$.
Consider now a JIT set $E'$ with $|E'|<k$.
Then, we have that
\begin{equation}\label{eqn: a2}
  \sum_{J_{ij}\in E'} w_{ij}= \sum_{J_{ij}\in E'} (T+x_{j})=|E'|T+ \sum_{J_{ij}\in E'} x_{j} \leq (k-1)T+T\leq kT,
\end{equation}
which contradicts the relation in eq.~(\ref{eqn: a1}).
Thus, $|E'|=k$.
Based on eq.~(\ref{eqn: a1}),
we can now conclude that
\begin{equation}\label{eqn: a3}
  \sum_{J_{ij}\in E'} x_{j} \geq B.
\end{equation}

The fact that $J_{kh+1} \in E$ implies that this job is scheduled during the time interval $((k+1)T, (k+1)T+1]$ on $M_2$,
which in turn implies that it starts not later than in time point $B$ on $M_1$.
The fact that $p^{(1)}_{kh+1}=(k+1)T - B \geq kT$ and that any job in $E'$ has a due date which is not greater than $kT$
implies that all jobs in $E'$ shall be scheduled prior to job $J_{kh+1}$ on $M_1$.
Thus, we have that
\begin{equation}\label{eqn: a4}
  \sum_{J_{ij}\in E'} p_{1j}=\sum_{J_{ij}\in E'} x_{j}\leq B.
\end{equation}

Based on eq.~(\ref{eqn: a3}) and eq.~(\ref{eqn: a4})
we conclude that $\sum_{J_{j}\in E'} x_{j}= B$;
this means that,
by setting $S=\{x_j|J_{ij}\in E'\}$,
we can obtain a solution for the \textsc{$k$SUM} with $|S|=k$ and $\sum_{x_i\in S}x_i=B$.
\end{proof}

\begin{figure}[t]
  \centering
  \includegraphics[scale=0.31]{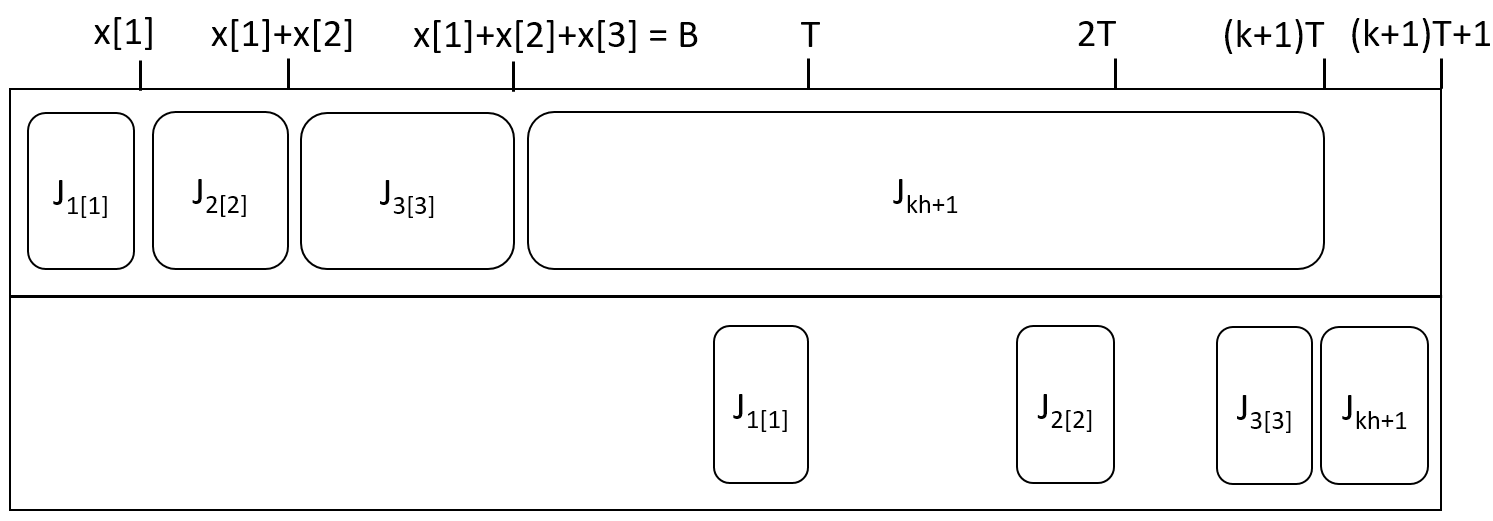}
\caption{An optimal schedule for a reduced instance as described in Theorem~\ref{theorem:one}, for $k = 3$.}
\label{figure:reduction-1}
\end{figure}

\subsection{An XP Algorithm for a Parameterized Number of Machines}

In this subsection we prove that the hardness result proven in Theorem~\ref{theorem:one} is tight,
by providing an XP algorithm for the \textsc{$F_m||\sum_{J_{j} \in E} w_j$} problem. We start by providing an XP algorithm for $m=3$ and then show how we can extend the result for a parameterized number of machines.

To do so, let us denote (without loss of generality) the different due dates by $d_1, \ldots, d_{\#d}$, such that $d_1 < d_2 < \ldots < d_{\#d}$.
Moreover, we say that job $J_j$ is of type $i$ if its due date is $d_i$,
for $i=1,\ldots,\#d$.

Importantly,
the fact that all jobs of the same type share the same due date implies that at most one of them can be scheduled in a JIT mode.

Let $\Omega$ be a set that includes all subsets of jobs sharing the property
that each $\omega\in \Omega$ includes at most $\#d$ jobs each of which is of a different type.
Note that
($i$) the set of feasible JIT jobs is an element of $\Omega$;
and
($ii$) $\mid\Omega\mid=O(n^{\#d})$.

Consider a specific subset of jobs $\omega\in \Omega$.
Next,
we explain how we can determine, in XP time,
whether there exists a feasible schedule where all jobs in $\omega$ are scheduled in a JIT mode.
Given a set $\omega$, we renumber jobs in the set $\omega=\{J_1,J_2,\ldots,J_{d'}\}$ such that $d_1<d_2<\ldots<d_{d'}$ ($d' \leq \#d$).
We first check if the condition that $d_{j-1}+p^{(3)}_j\leq d_{j}$ holds for $j=1,\ldots,d'$, where $d_0=0$ by definition.
If not, then there is no feasible schedule with $E=\omega$.
Otherwise, based on Lemma~\ref{lemma:Choi1},
we have that there are $d'!$ possible permutations to schedule jobs in $\omega$ on machines $M_1$ and $M_2$
(as according to this lemma there exists an optimal schedule where the permutation on the first two machines is identical).
For each such possible permutation $\sigma$, we schedule jobs on $M_1$ and $M_2$ according to Lemma~\ref{lemma:sh},
and therefore obtain (in linear time) the completion time $C^{(2)}_j$ of each job $J_j$ on $M_2$.

Finally, we check if job $J_j$ is indeed ready to be scheduled during the time interval $(d_{j}-p^{(3)}_j, d_{j}]$ on $M_3$ for $j=1,\ldots,n$;
that is, if $C^{(2)}_j\leq d_{j}-p^{(3)}_j$, for $j=1,\ldots,n$.
If this last condition holds,
then there is a feasible schedule with $E=\omega$.
By checking, in a similar fashion, the feasibility of each subset in $\Omega$,
we can find (in XP time) an optimum schedule.
To summarize the above discussion, we can use Algorithm~\ref{algorithm:xp-for-d} below to solve the \textsc{$F_3||\sum_{J_{j} \in E} w_j$} problem in XP time wrt.\ $\#d$.

\begin{algorithm}
\setstretch{\setstretchvalue}
  \textbf{Initialization:} \\
    Set $\mathrm{Opt} = 0$, $d_0=0$, and $E^*=\omega^*=\emptyset.$ \\
  \textbf{Step 1:} \\
    Construct a set $\Omega$, which includes all possible subsets of jobs that are of different types. \\
  \textbf{Step 2:} \\
   \While {$\Omega\neq \emptyset$} {
    \textbf{Step 2.1:} \\
    Select a specific set $\omega\in \Omega$, and set $\Omega=\Omega\setminus \omega$. \\
    Renumber jobs in $\omega = \{J_1, J_2, \ldots, J_{d'}\}$ such that $d_1<d_2<\ldots<d_{d'}$. \\
     If $\sum_{J_j\in \omega}w_j\leq \mathrm{Opt}$ {return to Step 2}. \\
     \textbf{Step 2.2:} \\
        \For {$j=1$ to $d'$} {
          \If {$d_{j - 1} + p^{(3)}_j > d_{j}$} {
            return to Step 2.
          }
        }
      \textbf{Step 2.3:} \\
        Construct a set $\Sigma$ which includes all possible permutations of the $d'$ jobs in $\omega$. \\
      \textbf{Step 2.4:} \\
       \While {$\Sigma\neq \emptyset$} {
          Select a specific permutation $\sigma\in \Sigma$, and set $\Sigma=\Sigma\setminus \sigma$. \\
          \textbf{Step 2.4.1:} \\
            Define $J_{[j]}$ as the $j$'th job in $\sigma$, and set
             $C^{(1)}_{[0]}=C^{(2)}_{[0]} = 0$. \\
          \textbf{Step 2.4.2:} \\
            \For {$j=1$ to $d'$} {
              Set $C^{(1)}_{[j]}=C^{(1)}_{[j-1]}+p^{(1)}_{[j]}$ and
               $C^{(2)}_{[j]}=\max\{C^{(2)}_{[j-1]},C^{(1)}_{[j]}\}+p^{(2)}_{[j]}$.
              \If {$C^{(2)}_j> d_{j}-p^{(3)}_j$} {
               return to Step 2.4.
              }
            }
        \textbf{Step 2.4.3:} \\
              Set $\mathrm{Opt} = \sum_{J_j\in \omega}w_j$, $E^*=\omega$, and $\sigma^*=\sigma$.
              Goto Step 2.
            }
        }
      \textbf{Output:} \\
      The optimal set of on-time jobs is $E^*$.
      Schedule jobs in $E^*$ on $M_1$ and $M_2$ according to Lemma~\ref{lemma:sh} with permutation $\sigma=\sigma^*$.
      Moreover,
      schedule each job $J_j\in E^*$  on $M_3$ during the time interval $(d_j-p_j^{(3)},d_j]$.
  \caption{An XP algorithm for the \textsc{$F_3||\sum_{J_{j} \in E} w_j$} problem wrt.\ $\#d$.}\label{algorithm:xp-for-d}
\end{algorithm}

\begin{theorem}\label{theorem:twotwo}
  The \textsc{$F_3||\sum_{J_{j} \in E} w_j$} problem is solvable in $O(n^{\#d+1}\#d!)$ time, and thus belongs to XP wrt.\ the parameter $\#d$.
\end{theorem}

\begin{proof}
The fact that Algorithm~\ref{algorithm:xp-for-d} solves the \textsc{$F_3||\sum_{J_{j} \in E} w_j$} problem
follows from the discussion that precedes the algorithm description.
Step 1 requires $O(n^{\#d})$ time, as is the number of times we perform Step 2.
The fact that the most time consuming step within Step 2 of Algorithm~\ref{algorithm:xp-for-d} is Step 2.4,
which requires $O(n\#d!)$ time,
completes our proof.
\end{proof}

We note that the result in Theorem~\ref{theorem:twotwo} can be extended to show that the more general \textsc{$F_m||\sum_{J_{j} \in E} w_j$} problem
is solvable in XP time for any fixed $m$ value.
The only change required in Algorithm~\ref{algorithm:xp-for-d} is in Step 2.3 where we need to enumerate not only all possible permutations on the first two machine,
but also the permutations on machines $M_3, M_4,\ldots,M_{m-1}$
(recall that according to Lemma~\ref{lemma:Choi1} and Lemma~\ref{lemma:Choi2} there exists an optimal schedule
in which job permutations on the first two machines are identical;
further, in an optimal schedule job permutation on the last machine follows the EDD rule).
We conclude as follows.

\begin{corollary}\label{corollary:one}
  The \textsc{$F_m||\sum_{J_{j} \in E} w_j$} problem is solvable in $O(n^{\#d+1}(\#d!)^{m-2})$ time,
  and thus,
  when $m$ is fixed,
  belongs to XP wrt.\ the number $\#d$ of different due dates.
\end{corollary}

\section{Two-Machines JIT Flow-Shop wrt.\ Combined Parameters}\label{section: pre2}

Theorem~\ref{theorem:one} shows the (fixed-parameter) intractability of the \textsc{$F_2||\sum_{J_j \in E} w_j$} problem
when we combine $\#d$ with the number of different processing times on the second machine ($\#p^{(2)}$).
Below, we analyze the tractability of the \textsc{$F_2||\sum_{J_j \in E} w_j$} problem
when we combine $\#d$ with either the number of different processing times on the first machine ($\#p^{(1)}$) or with the number of different weights ($\#w$).

Since we focus on the two machine case, according to Lemmas~\ref{lemma:Choi1} and \ref{lemma:Choi2},
there exists an optimal schedule with all jobs in the set $E$ being scheduled according to the EDD rule on both machines.
Suppose that all $n$ jobs are numbered according to the EDD rule such that $d_1\leq d_2\leq\ldots\leq d_n$.
Moreover,
consider two subschedules $S_1$ and $S_2$ defined on a subset of the first $j$ jobs,
with a corresponding feasible JIT sets $E_1$ and $E_2$, respectively.
Then the following lemma holds (see Shabtay~\cite{DBLP:journals/eor/Shabtay12}).

\begin{lemma}\label{lemma:shab1}
  Subschedule $S_2$ is dominated by subschedule $S_1$ if the following two conditions holds:
  ($i$) $\sum_{J_{j} \in {E_1}} w_j \geq \sum_{J_{j} \in {E_2}} w_j$;
  and
  ($ii$) $\sum_{J_{j} \in {E_1}} p^{(1)}_{j} \leq \sum_{j \in {E_2}} p^{(1)}_{j}$.
\end{lemma}

\subsection{FPT for the Two-Machine Case wrt.\ $\#d+\#p^{(1)}$}\label{section: pre2.1}

In this subsection, we say that two jobs $J_i$ and $J_j$ are of the same \emph{type} if both $d_i=d_j$ and $p^{(1)}_i=p^{(1)}_j$.
Let $k$ be the number of different job types ($k\leq \#d \cdot \#p^{(1)}$).
Note that
($i$) at most one job of each type can be scheduled in a JIT mode;
and
($ii$) jobs of the same type are differentiated by their weight and by their processing time on the second machine.

Given an instance for the \textsc{$F_2||\sum_{J_{j} \in E} w_j$} problem,
we partition the set $\mathcal{J}$ into subsets of job types $\mathcal{J}_1, \mathcal{J}_2,\ldots,\mathcal{J}_k$,
where all jobs in $\mathcal{J}_i$ ($i=1,\ldots,k$) have the same due date, $d_i$, and the same processing time on $M_1$, $p_i^{(1)}$.
Now,
let $J_{ij}$ be the $j$th job in $\mathcal{J}_i$.
Given a job schedule, we say that job set (type) $\mathcal{J}_i$ is a \emph{JIT job set} if one of jobs in $\mathcal{J}_i$ is scheduled in a JIT mode.
Accordingly,
a subset of job types is considered to be a feasible subset of job types
if there exists a feasible schedule which includes a single job of each of these types in the set of JIT jobs (set $E$).

Let $\Omega$ be a set that includes all possible (feasible or not) subsets of job types having different due dates.
Note that
($i$) the sets of all feasible subsets of job types is a subset of $\Omega$;
and
($ii$) $\mid\Omega\mid=O(2^k)$.

Consider now a specific subset of job sets $\omega\in \Omega$.
Below we explain how we can determine, in polynomial time,
the best feasible schedule (if such exists) out of all feasible schedules sharing the common property that all job sets in $\omega$ are JIT job sets.

Given a set $\omega$, we renumber job sets in $\omega$ such that $d_1<d_2<\ldots<d_{k'}$ ($k' \leq k$).
Consider now a feasible non-denominated subschedule $S$ defined on job sets $\mathcal{J}_1, \mathcal{J}_2,\ldots,\mathcal{J}_{l-1}$ ($l-1<k'$).
By definition, we have that the completion time of subschedule $S$ is
($i$) $\sum_{i=1}^{l-1} p^{(1)}_i$ on $M_1$;
and
($ii$) $d_{l-1}$ on $M_2$.
Consider now the extension of $S$ to a subschedule on the extended set that includes also $\mathcal{J}_{l}$.
In such an extension, job $J_{lj}\in \mathcal{J}_l$ is a nominee to be assigned to the set $E$
only if $\max \{\sum_{i=1}^{l} p^{(1)}_i, d_{l-1}\}+p_{lj}^{(2)}\leq d_l$.
Let $N(l)$ be the set of all jobs of type $l$ that are nominees to be assigned to the set $E$.
If $N(l)=\emptyset$,
then we have no feasible solution with all job types in $\omega$ being JIT.
Otherwise, based on Lemma~\ref{lemma:shab1},
we have that it is optimal to assign to the set $E$ job with the smallest processing time on $M_1$ among all jobs in $N(l)$.
Based on this observation and the above analysis, we can use Algorithm~\ref{algorithm:fpt-for-dp1} below to solve the \textsc{$F_2||\sum_{J_{j} \in E} w_j$} problem with respect to the combined parameter $\#d+\#p^{(1)}$.

\begin{algorithm}
\setstretch{\setstretchvalue}
  \textbf{Initialization:} \\
    Set $\mathrm{Opt} = 0$, $d_0=0$, and $E^*=\emptyset$. \\
  \textbf{Step 1:} \\
    Partition the set $\mathcal{J}$ into subsets of job types $\mathcal{J}_1, \mathcal{J}_2,\ldots,\mathcal{J}_k$,
    where all jobs in $\mathcal{J}_i$ have the same due date, $d_i$, and the same processing time on $M_1$, $p_i^{(1)}$. \\
  \textbf{Step 2:} \\
    Construct a set $\Omega$, which includes all possible subsets of job types having different due dates. \\
  \textbf{Step 3:} \\
   \While {$\Omega\neq \emptyset$} {
    \textbf{Step 3.1:} \\
      Select a specific set $\omega\in \Omega$. Set $\Omega=\Omega\setminus \omega$. \\
      Let $k'$ be the number of job sets in $\omega$. \\
      Renumber job sets in $\omega$ such that $d_1<d_2<\ldots<d_{k'}$. \\
      Set $V(\omega)=0$ and $E(\omega)=\emptyset$. \\
      \textbf{Step 3.2:} \\
        Set $P_1=0$ \\
        \For {$i=1$ to $k'$} {
          Set $P_1=P_1+p_i^{(1)}$. \\
          Set $N(i)=\{J_{ij}\in \mathcal{J}_i \mid \max \{P_1, d_{i-1}\}+p_{ij}^{(2)}\leq d_{i}\}$. \\
          If $N(i)=\emptyset$  {return to Step 3.} \\ {
          Set $j^*=\arg \max \{w_{ij} | J_{ij}\in N(i)\}$. \\
          Set $E(\omega)=E(\omega) \bigcup \{J_{ij^*}\}$; and $V(\omega)=V(\omega)+w_{ij^*}$.  \\
          }
        }
      \textbf{Step 3.3:} \\
        \If {$V(\omega)> \mathrm{Opt}$} {
          Set $\mathrm{Opt}=V(\omega)$ and $E^*=E(\omega)$. \\
        }
        Goto Step 2.
      }
      \textbf{Output:} \\
        The optimal set of on-time jobs is $E^*$.
        Schedule jobs in $E^*$ on $M_1$ one after the other with no delays, and on $M_2$ during the time interval $(d_j-p_j^{(2)},d_j]$.
  \caption{A fixed-parameter algorithm for the \textsc{$F_2||\sum_{J_{j} \in E} w_j$} problem wrt.\ $\#d$+$\#p^{(1)}$.}\label{algorithm:fpt-for-dp1}
\end{algorithm}

\begin{theorem}\label{theorem:three}
  Algorithm~\ref{algorithm:fpt-for-dp1} solves the \textsc{$F_2||\sum_{J_{j} \in JIT} w_j$} problem in $O(n2^k)$ time,
  where $k=\#d \cdot \#p^{(1)}$.
  Thus, the problem is fixed-parameter tractable when parameterized by combining
  the number $\#d$ of different due dates
  with the number $\#p^{(1)}$ of different processing times on the first machine.
\end{theorem}

\begin{proof}
The fact that Algorithm~\ref{algorithm:fpt-for-dp1} solves the \textsc{$F_2||\sum_{J_{j} \in E} w_j$} problem
follows from the discussion that precedes the algorithm description.
Step 1 assigns jobs to the different job sets (types) and requires a linear time.
In Step 2 we construct a set $\Omega$ which includes all possible subsets of job types having different due dates;
this step requires $O(2^k)$ time (as $|\Omega|=O(2^k)$).
In Step 3, for each $\omega\in \Omega$, our algorithm finds the best feasible schedule
(if such exists)
out of all schedules where all job sets in $\omega$ are JIT job sets.
The fact that $|\Omega|=O(2^k)$ and that the most time consuming operation within Step 3 is Step 3.2,
which requires $O(n)$ time,
completes our proof.
\end{proof}

\subsection{FPT for the Two-Machine Case wrt.\ $\#d+\#w$}\label{section: pre2.2}

In this subsection we present an algorithm,
following a similar approach to that used in subsection~\ref{section: pre2.1},
which is a fixed-parameter algorithm for the case where we combine the two parameters $\#d$ and $\#w$.
Here we say that two jobs $J_i$ and $J_j$ are of the same \emph{type} if both $d_i=d_j$ and $w_i=w_j$.
Let $k$ be the number of different job types ($k\leq \#d \cdot \#w$).

Given an instance for the \textsc{$F_2||\sum_{J_{j} \in E} w_j$} problem,
we partition the set $\mathcal{J}$ into subsets of job types $\mathcal{J}_1, \mathcal{J}_2,\ldots,\mathcal{J}_k$,
where all jobs in $\mathcal{J}_i$ ($i=1,\ldots,k$) have the same due date, $d_i$, and the same weight, $w_i$.
Now, let $J_{ij}$ be the $j$th job in $\mathcal{J}_i$.
Given a job schedule, we say that job set (type) $\mathcal{J}_i$ is a \emph{JIT job set} if one of jobs in $\mathcal{J}_i$ is scheduled in a JIT mode.
Accordingly, a subset of job sets (types) is considered to be a feasible subset of job sets if there exists a feasible schedule
which includes a single job of each of these sets in the set of JIT jobs (i.e., the set $E$).

Now, let $\Omega$ be a set that includes all possible (feasible or not) subsets of job types having different due dates.
Note that
($i$) the sets of all feasible subsets of job types is a subset of $\Omega$;
and
($ii$) $\mid\Omega\mid=O(2^k)$.
Consider a specific subset of job sets $\omega\in \Omega$.
Below we explain how we can determine, in polynomial time, the best feasible schedule (if such exists)
out of all the feasible schedules sharing the common property that all job sets in $\omega$ are JIT job sets.

Given a set $\omega$,
we renumber job types in $\omega$ such that $d_1<d_2<\ldots<d_{k'}$ ($k' \leq k$).
Consider now a feasible non-denominated subschedule $S$ defined on job sets $\mathcal{J}_1, \mathcal{J}_2,\ldots,\mathcal{J}_{l-1}$ ($l-1<k'$).
By definition,
we have that
($i$) the completion time of the subschedule $S$ on $M_2$ is at time $d_{l-1}$;
and
($ii$) the total weight of the JIT jobs in $S$ is $\sum_{i=1}^{l-1} w_i$.
Now, let $P^{(1)}(S)$ be the completion time of the subschedule $S$ on $M_1$,
and consider the extension of $S$ to a subschedule on the extended set of job types that includes also $\mathcal{J}_{l}$.
In such an extension, job $J_{lj}\in \mathcal{J}_l$ is a nominee to be assigned to the set $E$ only
if $\max \{P^{(1)}(S)+p_{lj}^{(1)}, d_{l-1}\}+p_{lj}^{(2)}\leq d_l$.
Let $N(l)$ be the set of all jobs of type $l$ that are nominees to be assigned to the set $E$.
If $N(l)=\emptyset$,
then we have no feasible solution with all job sets in $\omega$ being JIT.
Otherwise, based on Lemma~\ref{lemma:shab1},
it follows that it is optimal to assign to the set $E$ job with the smallest processing time on $M_1$ among all jobs in $N(l)$.
Based on this observation and the above analysis,
we present Algorithm~\ref{algorithm:fpt-for-dp2} to solve the \textsc{$F_2||\sum_{J_{j} \in E} w_j$} problem with respect to the combined parameter $\#d$+$\#w$.

\begin{algorithm}
\setstretch{\setstretchvalue}
  \textbf{Initialization:} \\
    Set $\mathrm{Opt} = 0$, $d_0=0$, and $E^*=\emptyset$. \\
  \textbf{Step 1:} \\
    Partition the set $\mathcal{J}$ into subsets of job types $\mathcal{J}_1, \mathcal{J}_2,\ldots,\mathcal{J}_k$,
    where all jobs in $\mathcal{J}_i$ have the same due date, $d_i$, and the same weight, $w_i$. \\
  \textbf{Step 2:} \\
    Construct a set $\Omega$, which includes all possible subsets of job types having different due dates. \\
  \textbf{Step 3:} \\
    \While {$\Omega\neq \emptyset$} {
      \textbf{Step 3.1:} \\
        Select a specific set $\omega\in \Omega$. Set $\Omega=\Omega\setminus \omega$. \\
        Let $k'$ be the number of job sets in $\omega$. \\
        Renumber job sets in $\omega$ such that $d_1<d_2<\ldots<d_{k'}$. \\
        Set $V(\omega)=0$ and $E(\omega)=\emptyset$. \\
     \textbf{Step 3.2:} \\
        Set $P_1=0$ \\
        \For {$i=1$ to $k'$} {
 %        Set $P_1=P_1+p_i^{(1)}$ \\
          Set $N(i)=\{J_{ij}\in \mathcal{J}_i \mid \max \{P_1+p_{ij}^{(1)}, d_{i-1}\}+p_{ij}^{(2)}\leq d_{i}\}$. \\
          If $N(i)=\emptyset$  {return to Step 3.} \\ {
          Set $j^*=\arg \min \{p^{(1)}_{ij} | J_{ij}\in N(i)\}$. \\
          Set $E(\omega)=E(\omega) \bigcup \{J_{ij^*}\}$; $P_1=P_1+J_{ij^*}$; and $V(\omega)=V(\omega)+w_{ij^*}$.  \\
          }
        }
      \textbf{Step 3.3:} \\
      \If {$V(\omega)> \mathrm{Opt}$} {
        Set $\mathrm{Opt}=V(\omega)$ and $E^*=E(\omega)$. \\
      }
      Goto Step 2.
    }
    \textbf{Output:} \\
      The optimal set of on-time jobs is $E^*$.
      Schedule jobs in $E^*$ on $M_1$ one after the other with no delays, and on $M_2$ during the time interval $(d_j-p_j^{(2)},d_j]$.
  \caption{A fixed-parameter algorithm for the \textsc{$F_2||\sum_{J_{j} \in E} w_j$} problem wrt.\ $\#d+\#w$.}\label{algorithm:fpt-for-dp2}
\end{algorithm}

\begin{theorem}\label{theorem:four}
  Algorithm~\ref{algorithm:fpt-for-dp2} solves the \textsc{$F_2||\sum_{J_{j} \in JIT} w_j$} problem in $O(n2^k)$ time,
  where $k=d_{\#d}\cdot \#w$.
  Thus,
  the problem is fixed-parameter tractable when parameterized by combining the number $\#d$ of different due dates
  with the number $\#w$ of different weights.
\end{theorem}

\begin{proof}
The fact that Algorithm~\ref{algorithm:fpt-for-dp2} solves the \textsc{$F_2||\sum_{J_{j} \in E} w_j$} problem
follows the discussion that precedes the algorithm description.
Step 1 assigns jobs to the different job sets and requires a linear time.
In Step 2 we construct a set $\Omega$ that includes all possible subsets of job types having different due dates;
this step requires $O(2^k)$ time.
In Step 3, for each $\omega\in \Omega$, our algorithm finds the best feasible schedule
(if such exists)
out of all schedules where all job sets in $\omega$ are JIT job sets.
The fact that $|\Omega|=O(2^k)$ and that the most time consuming operation within Step 3 is Step 3.2,
which requires $O(n)$ time,
completes our proof.
\end{proof}

\section{Three-Machines JIT Flow-Shop wrt.\ Combined Parameters}\label{section: pre3}

In Section~\ref{section: pre2} we prove that the two-machine problem is FPT with respect to both $\#d$+$\#p^{(1)}$ and $\#d$+$\#w$.
Next we show that these two cases become W[1]-hard when we consider JIT flow-shop scheduling on three machines.
Specifically,
the parameterized intractability of the \textsc{$F_3||\sum_{J_j \in E} w_j$} problem with respect to $\#d$+$\#p^{(1)}$ will be proven using the following lemma.

\begin{lemma} \label{lemma:Intract}
  The \textsc{$F_3||\sum_{J_j \in E} w_j$} problem with unit processing times on the $M_1$ reduces to the \textsc{$F_2||\sum_{J_j \in E} w_j$} problem.
\end{lemma}

\begin{proof}
Given an instance $I=\{n, p_j^{(1)},p_j^{(2)},w_j, d_j\}$ for the \textsc{$F_2||\sum_{J_j \in E} w_j$} problem,
we construct the following instance $\bar{I}=\{\bar{n}, \bar{p}_j^{(1)},\bar{p}_j^{(2)}, \bar{p}_j^{(3)}, \bar{w}_j, \bar{d}_j\}$
for the \textsc{$F_3||\sum_{J_j \in E} w_j$} problem.
We set $\bar{n} = n$.
Moreover,
for $j=1,\ldots,\bar{n}$,
we set
($i$) $\bar{p}_j^{(1)}=1$;
($ii$) $\bar{p}_j^{(2)}=p_j^{(1)}$;
($iii$) $\bar{p}_j^{(3)}=p_j^{(2)}$;
($iv$) $\bar{w}_j=w_j$;
and
($v$) $\bar{d}_j=d_j+1$.
Then, it follows that that there exists a feasible schedule for the constructed instance of the \textsc{$F_3||\sum_{j \in E} w_j$} problem
with all jobs in the set $E$ being scheduled in a JIT mode;
\textit{iff} the same set of jobs is a feasible set of jobs for the corresponding instance of the \textsc{$F_2||\sum_{J_j \in E} w_j$} problem.
\end{proof}

We conclude the following,
based on Theorem~\ref{theorem:one} and on Lemma~\ref{lemma:Intract}.

\begin{theorem}\label{theorem:hard2}
  The \textsc{$F_3||\sum_{J_j \in E} w_j$} problem is W[1]-hard when parameterized by $\#d$+$\#p^{(1)}$
  \textit{even}
  if all jobs have unit processing times on $M_3$.
\end{theorem}

The next theorem shows that the \textsc{$F_3||\sum_{J_j \in E} w_j$} problem is W[1]-hard as well when we combine $\#d$ and $\#w$.

\begin{theorem}\label{theorem:five}
  The \textsc{$F_3||\sum_{J_j \in E} w_j$} problem is W[1]-hard when parameterized by $\#d$+$\#w$.
\end{theorem}

\begin{proof}
We prove that the \textsc{$F_3||\sum_{J_{j} \in E} w_j$} problem is W[1]-hard with respect to $\#d$+$\#w$
by providing a polynomial reduction from the \textsc{$k$SUM} problem (see Definition~\ref{definition:ksum}).
Given an instance of \textsc{$k$SUM} we construct an instance for the decision version of the $F_3 || \sum_{J_{j}\in E} w_{j}$ problem as follows.
The set $\mathcal{J}$ includes $n = kh+2$ jobs.
The first $kh$ jobs are the union of $k$ sets $\mathcal{J}_1, \mathcal{J}_2,\ldots,\mathcal{J}_k$,
where each set $\mathcal{J}_i$ ($i\in \{1,\ldots,k\}$) includes $h$ jobs.
Let $J_{ij}$ be the $j$th job in the set $\mathcal{J}_i$.
Moreover, for any job $J_{ij}$, let $d_{ij}$, $w_{ij}$, $p^{(1)}_{ij}$, and $p^{(2)}_{ij}$,
be the due date, the weight, and the processing times on machines $M_1$ and $M_2$, respectively.

For any job $J_{ij}$,
we set
($i$) $d_{ij}=kT+i+1$, where $T=\sum_{x_i\in X}x_i$;
($ii$) $w_{ij}=T$;
($iii$) $p^{(1)}_{ij}=x_j$;
($iv$) $p^{(2)}_{ij}=T-x_j$;
and
($v$) $p^{(3)}_{ij}=1$.
For job $J_{kh+1}$,
we set
($i$) $d_{kh+1}=B+2$;
($ii$) $w_{kh+1}=k^2(T+1)^2$;
($iii$) $p^{(1)}_{kh+1}=p^{(3)}_{kh+1}=1$;
and
($iv$) $p^{(2)}_{kh+1}=B$.
For job $J_{kh+2}$,
we set
($i$) $d_{kh+2}=2kT+2$;
($ii$) $w_{kh+2}=k^2(T+1)^2$;
($iii$) $p^{(1)}_{kh+2}=kT-B$;
($iv$) $p^{(2)}_{kh+2}=kT$;
and
($v$) $p^{(3)}_{kh+2}=1$.
In our decision version of the $F_m || \sum_{J_{j}\in E} w_{j}$ problem
we ask whether there is a feasible schedule (partition) with $\sum_{J_{j}\in E} w_{j}\geq kT+2k^2(T+1)^2$.
This finishes the description of the polynomial-time reduction
(a tabular representation of jobs created by the reduction is given in Table~\ref{table:reduction-2-general}).

\begin{table}[t]
\centering
\begin{tabular}{ c | c c c c c}
            & $p^{(1)}$    & $p^{(2)}$    & $p^{(3)}$     & $w$           & $d$             \\ \hline
$J_{11}$    & $x_1$        & $T - x_1$    & $1$           & $T$           & $kT + 2$        \\
$\vdots$    & $\vdots$     & $\vdots$     & $\vdots$      & $\vdots$      & $\vdots$        \\
$J_{1h}$    & $x_h$        & $T - x_h$    & $1$           & $T$           & $kT + 2$        \\ \hline
$\vdots$    & $\vdots$     & $\vdots$     & $\vdots$      & $\vdots$      & $\vdots$        \\ \hline
$J_{k1}$    & $x_1$        & $T - x_1$    & $1$           & $T$           & $kT + k + 1$    \\
$\vdots$    & $\vdots$     & $\vdots$     & $\vdots$      & $\vdots$      & $\vdots$        \\
$J_{kh}$    & $x_h$        & $T - x_h$    & $1$           & $T$           & $kT + k + 1$    \\ \hline
$J_{kh+1}$  & $1$          & $B$          & $1$           & $k^2(T+1)^2$  & $B + 2$         \\ \hline
$J_{kh+2}$  & $kT - B$     & $kT$         & $1$           & $k^2(T+1)^2$  & $2kT + 2$
\end{tabular}
\caption{jobs created in the reduction described in the proof of Theorem~\ref{theorem:five}.}
\label{table:reduction-2-general}
\end{table}

Notice that,
since all jobs in $\mathcal{J}_i$ ($i=1,\ldots,k$) share the same due date of $kT+i+1$,
it follows that,
in any feasible schedule,
at most one of them can be scheduled in a JIT mode.
Therefore, we have that
\begin{equation}\label{eqn: yy1}
  |E|\leq k+2.
\end{equation}

Let us begin by proving that,
if we have a yes-instance for the \textsc{$k$SUM},
then there exists a feasible schedule for the constructed instance of the $F_3 || \sum_{J_{j}\in E} w_{j}$ problem
with $\sum_{J_{j}\in E} w_{j}\geq kT+2k^2(T+1)^2$.
The fact that we have a yes-instance to \textsc{$k$SUM} implies that there exists a subset of $k$ elements, \linebreak
$S = \{x_{[1]},x_{[2]},\ldots,x_{[k]}\} \subseteq X$,
such that $\sum_{j=1}^k x_{[j]} = B$, where $[j]$ is the index of the $j$th element in $S$.

We construct the following schedule for the constructed instance of our scheduling problem.
We set $E=\{J_{1[1]},J_{2[2]},\ldots,J_{k[k]}\} \cup \{J_{kh+1}\} \cup \{J_{kh+2}\}$ and $T=\mathcal{J} \backslash E$.
Moreover, we maintain the same processing order on the entire set of machines;
specifically, $\sigma_1=\sigma_2=\sigma_3= \{J_{kh+1}, J_{1[1]},J_{2[2]},\ldots,J_{k[k]},J_{kh+2}\}$.
Following this processing order,
we schedule job $J_{kh+1}$ during the time interval $(0,1]$ on $M_1$;
during the time interval $(1,B+1]$ on $M_2$;
and during the time interval $(B+1,B+2]$ on $M_3$.
Then, for $i=1,\ldots,k$,
we schedule job $J_{i[i]}$ during the time interval $(1+ \sum_{j=1}^{i-1} x_{[j]}, 1+ \sum_{j=1}^{i} x_{[j]}]$ on $M_1$;
during the time interval $(B+1+(i-1)T-\sum_{j=1}^{i-1} x_{[j]}, B+1+iT-\sum_{j=1}^{i} x_{[j]}]$ on $M_2$;
and during the time interval $(kT+i,kT+i+1]$ on $M_3$.
Finally,
we schedule job $J_{kh+2}$ during the time interval $(B+1,kT+1]$ on $M_1$;
during the time interval $(kT+1,2kT+1]$ on $M_2$;
and during the time interval $(2kT+1,2kT+2]$ on $M_3$.
The constructed schedule is illustrated in Figure~\ref{figure:reduction-2}.

\begin{figure}[t]
  \centering
  \includegraphics[scale=0.31]{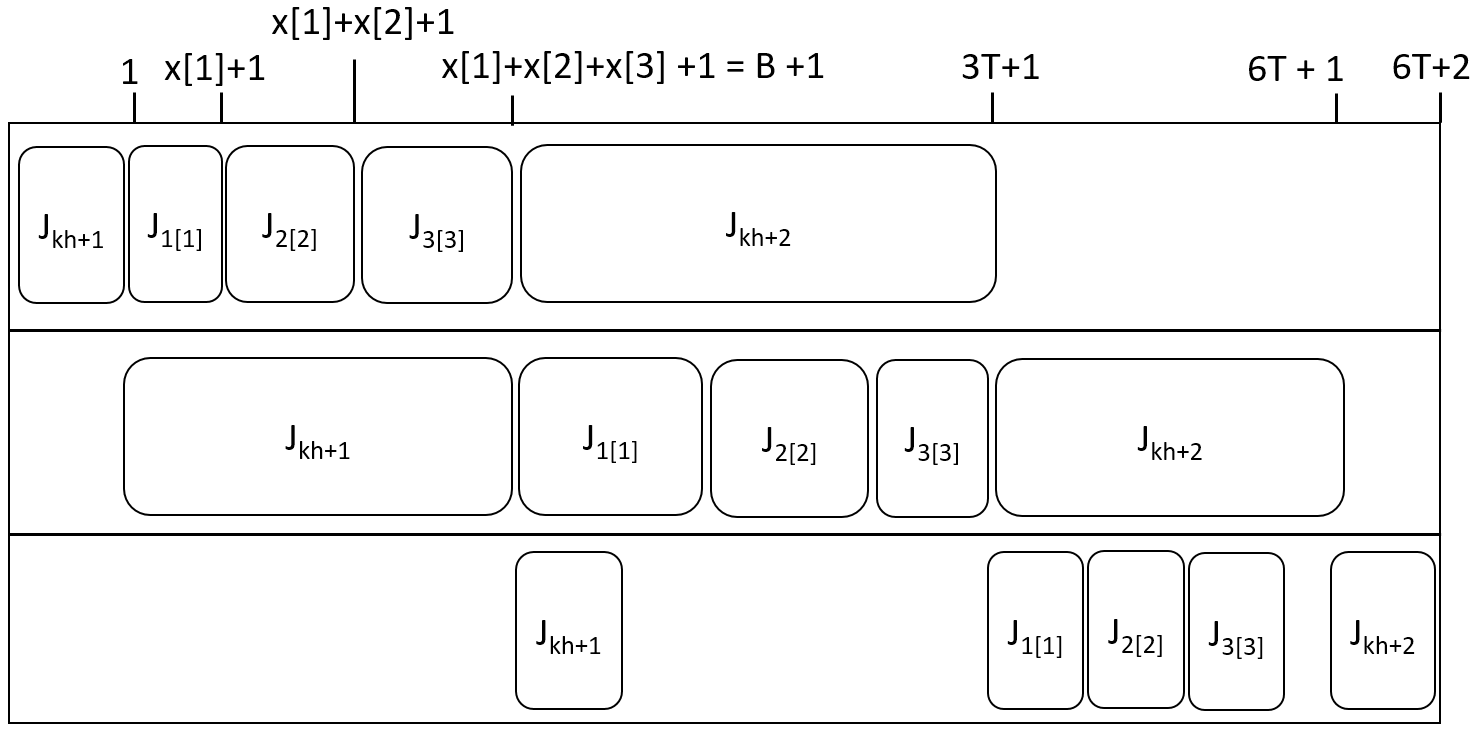}
\caption{An optimal schedule for a reduced instance, as described in Theorem~\ref{theorem:five}, for $k = 3$.}
\label{figure:reduction-2}
\end{figure}

We next prove that the constructed schedule is feasible (with all jobs in the set $E$ being scheduled in a JIT mode).
To this end,
we shall prove that
(a) there is no overlap between processing operations on each one of the three machines;
and (b) there is no overlap between processing operations of each of jobs on the different machines.

Let us first prove (a).
The fact that job $J_{1[1]}$ starts its processing on $M_1$ right after the completion of job $J_{kh+1}$ at time $1$;
that job $J_{i+1[i+1]}$ starts its processing on $M_1$ right after the completion of job $J_{i[i]}$ at time $1+ \sum_{j=1}^{i} x_{[j]}$ for $i=1,\ldots,k-1$;
and that job $J_{kh+2}$ starts its processing on $M_1$ right after the completion of job $J_{k[k]}$ at time $1+ \sum_{j=1}^{k} x_{[k]}=1+B$,
implies that there is no overlap of processing operations on $M_1$.
The fact that job $J_{1[1]}$ starts its processing on $M_2$ right after the completion of job $J_{kh+1}$ at time $B+1$;
that job $J_{i+1[i+1]}$ starts its processing on $M_2$ right after the completion of job $J_{i[i]}$ at time $B+1+iT-\sum_{j=1}^{i} x_{[j]}$ for $i=1,\ldots,k-1$;
and that job $J_{kh+2}$ starts its processing on $M_2$ right after the completion of job $J_{k[k]}$ at time $B+1+kT-\sum_{j=1}^{k} x_{[k]}=kT+1$,
implies that there is no overlap of processing operations on $M_2$.
Then,
the fact that there is no overlap between processing operations on $M_3$ follows from the fact that
($i$) the start time of $J_{1[1]}$ on $M_3$ is at time $kT+1$ which is not earlier than the completion time of $J_{kh+1}$ on $M_3$ at time $B+2$ (as $T > B$);
($ii$) job $J_{i+1[i+1]}$ starts its processing on $M_3$ right after the completion of job $J_{i[i]}$ at time $Tk+i+1$ for $i=1,\ldots,k-1$;
and
($iii$) job $J_{kh+2}$ starts its processing on $M_3$ at time $2kT+1$ which is later than the completion time of $J_{k[k]}$ on $M_3$ at time $kT+k+1$ (as $T > k$).

We move on to prove (b).
This claim follows from the fact that
($i$) the processing of job $J_{kh+1}$ starts at $M_i$ exactly after its completion on $M_{i-1}$ for $i=2,3$;
($ii$) the processing of each of job in the set $\{J_{1[1]},J_{2[2]},\ldots,J_{k[k]}\}$
starts at $M_i$ only after the entire set is completed on $M_{i-1}$ for $i=2,3$;
and
($iii$) the processing of job $J_{kh+2}$ starts at $M_i$ exactly after its completion on $M_{i-1}$ for $i=2,3$.

The fact that the schedule is feasible and and all jobs in the set $E$ are completed in a JIT mode implies that
\begin{equation}\label{eqn: y1}
  \sum_{J_{j}\in E} w_{j}= \sum_{j=1}^k w_{j[j]}+ w_{kh+1}+w_{kh+2}= kT+2k^2(T+1)^2;
\end{equation}
thus, we have a yes answer for the constructed scheduling instance.

Now we prove that,
if we have a no-instance for the \textsc{$k$SUM},
then a feasible schedule for the constructed instance of the $F_3 || \sum_{J_{j}\in E} w_{j}$ problem
with $\sum_{J_{j}\in E} w_{j}\geq kT+2k^2(T+1)^2$ does not exist.
By contradiction, assume that a feasible partition (schedule) $\tau = E\cup T$ with $\sum_{J_{j}\in E} w_{j}\geq kT+2k^2(T+1)^2$ exists.
The fact that at most a single job from each $\mathcal{J}_i$ can be scheduled in a JIT mode in any feasible schedule
implies that the total weight of the subset of JIT jobs among all jobs in $\cup_{i=1}^{k} \mathcal{J}_i$ is at most $kT$.
Therefore, both ${J_{kh+1}}$ and ${J_{kh+2}}$ are included in the set $E$
(as otherwise, we would have that $\sum_{J_{j}\in E} w_{j} \leq kT+k^2(T+1)^2< kT+2k^2(T+1)^2$,
contradicting our assumption that $\sum_{J_{j}\in E} w_{j}\geq kT+2k^2(T+1)^2$).

Let $E'=E \setminus \{J_{kh+1},J_{kh+2}\}$.
Since $\sum_{J_{j}\in E} w_{j}\geq kT+2k^2(T+1)^2$ and $\{J_{kh+1},J_{kh+2}\}\in E$, we conclude that
\begin{equation}\label{eqn: aa1}
  \sum_{J_{j}\in E'} w_{j} = \sum_{J_{j}\in E} w_{j} - w_{kh+1}-w_{kh+2} \geq kT.
\end{equation}

Based on eq.~(\ref{eqn: yy1}), eq.~(~\ref{eqn: aa1}), and the fact that $|E'|=|E-2|$, we conclude that $|E'|=k$.

The fact that ${J_{kh+1}}\in E$ implies that it has to be scheduled during the time interval $(0,1]$ on $M_1$;
during the time interval $(1,B+1]$ on $M_2$;
and during the time interval $(B+1,B+2]$ on $M_3$.
Next we prove that all the $k$ jobs in $E'$ shall be scheduled before $J_{kh+2}$ on ($i$) $M_2$ and ($ii$) $M_1$.
Let us first prove ($i$).
By contradiction, assume that one of jobs in $E'$ (say job $J_{ij}$) is scheduled after $J_{kh+2}$ on $M_2$.
The fact that we have to schedule $J_{kh+1}$ during the time interval $(1,B+1]$ on $M_2$
and that $p^{(2)}_{kh+2}=kT$ implies that job $J_{ij}$ will start its processing on $M_2$ not earlier than on time $kT+B+1$,
and thus will be completed later than its due date on $M_3$,
contradicting its containment in $E'$.
Let us now prove ($ii$).
By contradiction,
assume that one of jobs in $E'$ (say job $J_{ij}$) is scheduled after $J_{kh+2}$ on $M_1$.
The fact that we have to schedule $J_{kh+1}$ during the time interval $(0,1]$ on $M_1$
and that $p^{(1)}_{kh+2}=kT-B$
implies that job $J_{ij}$ will complete its processing on $M_1$ not earlier than on time $kT-B+x_i+1$,
and thus will be completed on $M_2$ not earlier than on time $(k+1)T-B+1$.
Due to ($i$) above, this further implies that $J_{kh+2}$ will be completed on $M_2$ not earlier than at time $(2k+1)T-B+1$
and on $M_3$ not earlier than at time $(2k+1)T-B+2>2kT+2=d_{kh+2}$,
contradicting the fact that $J_{kh+2}$ is an early job.

The fact that $J_{kh+2}\in E$ implies that $J_{kh+2}$ has to start not later than on time $B+1$ on $M_1$ and not later than on time $kT+1$ on $M_2$.
Based on ($i$),
we conclude that the start time of $J_{kh+2}$ on $M_2$ is not earlier than at time $p^{(1)}_{kh+1}+\sum_{J_{ij}\in E'} p^{(1)}_{ij}=1+\sum_{J_{ij}\in E'} x_{i}$. Therefore,
we have that $1+\sum_{J_{ij}\in E'} x_{i}\leq B+1$,
i.e.,
that $\sum_{J_{ij}\in E'} x_{i}\leq B$.
Based on ($i$) and ($ii$),
we also have that the start time of $J_{kh+2}$ on $M_1$ is not earlier than at time
$p^{(1)}_{kh+1}+ \max \{\sum_{J_{ij}\in E'} p^{(1)}_{ij} + p^{(1)}_{kh+2}, p^{(2)}_{kh+1}+ \sum_{J_{ij}\in E'} p^{(2)}_{ij}\} =
1+kT+ \max \{\sum_{J_{ij}\in E'} x_i -B, B - \sum_{J_{ij}\in E'} x_i \}$.
Therefore, we have that $1 + kT + \max \{\sum_{J_{ij}\in E'} x_i - B, B - \sum_{J_{ij}\in E'} x_i \} \leq kT + 1$,
which implies that $\sum_{J_{ij}\in E'} x_i = B$.
Thus, by setting $S=\{x_j|J_{ij}\in E'\}$, we can obtain a solution for the \textsc{$k$SUM} with $|S|=k$ and $\sum_{x_i\in S}x_i=B$.
\end{proof}

\section{Summary and Future Research}

In this paper we provide a parameterized analysis of the NP-hard JIT flow-shop scheduling problem on two and three machines.
The main parameter being studied is the number of different due dates.
We prove that the problem is intractable in the parameterized sense,
even when the scheduling is done on two machines,
and belongs to XP when the scheduling is done on any fixed number of $m\geq 2$ machines.

Then, we show that,
when combining the number of different due dates with either the number of different processing times on the first machine,
or the number of different weights,
the problem becomes fixed-parameter tractable when the scheduling is done on two machines.
It remains W[1]-hard,
however,
when the scheduling is done on three machines
(our results are summarized in Table~\ref{table:results}).

One immediate direction for future research is to focus on other parameters,
and to study the parameterized complexity of the JIT flow-shop problem with respect to those parameters.
Finally,
one can move to other machine environments such as unrelated machines, job-shop, and open-shop.

\section*{Acknowledgments}

This research was partially supported by Grant No. 2016049 from the United States-Israel Binational Science Foundation (BSF).

\bibliographystyle{plain}
\bibliography{bib}

\iffalse

JUNK

\begin{table}[t]
\centering
%
\begin{tabular}{ c | c c c c}
         & $p^1$    & $p^2$    & $w$      & $d$ \\ \hline
$J^1_1$  & $2$      & $1$      & $2$      & $25$ \\
$J^1_2$  & $3$      & $1$      & $3$      & $25$ \\
$J^1_3$  & $5$      & $1$      & $5$      & $25$ \\
$J^1_4$  & $7$      & $1$      & $7$      & $25$ \\
$J^1_5$  & $8$      & $1$      & $8$      & $25$ \\ \hline
$J^2_1$  & $2$      & $1$      & $2$      & $50$ \\
$J^2_2$  & $3$      & $1$      & $3$      & $50$ \\
$J^2_3$  & $5$      & $1$      & $5$      & $50$ \\
$J^2_4$  & $7$      & $1$      & $7$      & $50$ \\
$J^2_5$  & $8$      & $1$      & $8$      & $50$ \\ \hline
$J^3_1$  & $2$      & $1$      & $2$      & $75$ \\
$J^3_2$  & $3$      & $1$      & $3$      & $75$ \\
$J^3_3$  & $5$      & $1$      & $5$      & $75$ \\
$J^3_4$  & $7$      & $1$      & $7$      & $75$ \\
$J^3_5$  & $8$      & $1$      & $8$      & $75$ \\ \hline
$J_*$    & $63$     & $1$      & $625$    & $76$ \\
\end{tabular}
%
\caption{jobs created in the reduction described in the proof of Theorem~\ref{theorem:one}
for the instance of \textsc{$k$SUM} where the goal is to find $3$ numbers from the integers $2, 3, 5, 7, 8$
whose sum is $12$. Notice that $T = 25$, $n = 5$, and $k = 3$.}
\label{table:reduction-1}
\end{table}

\fi

\end{document}